\documentclass{article}
\usepackage{style}

\newcommand{\U}{\mathrm{U}}
\newcommand{\SU}{\mathrm{SU}}
\newcommand{\ii}{\mathrm{i}}
\newcommand{\diag}{\mathrm{diag}}
\newcommand{\anc}{\mathsf{anc}}

\title{Unitary synthesis with fewer T gates}
\author{Xinyu Tan\thanks{Google Quantum AI, Venice, CA 90291.} \thanks{Department of Mathematics, MIT, Cambridge, MA, 02139. Email: \texttt{norahtan@mit.edu}}}
\date{}

\begin{document}

\maketitle

\begin{abstract}
    We present a simple algorithm that implements an arbitrary $n$-qubit unitary operator using a Clifford+T circuit with T-count $O(2^{4n/3} n^{2/3})$. 
    This improves upon the previous best known upper bound of $O(2^{3n/2} n)$, while the best known lower bound remains $\Omega(2^n)$.
    Our construction is based on a recursive application of the cosine-sine decomposition, together with a generalization of the optimal diagonal unitary synthesis method by Gosset, Kothari, and Wu \cite{GKW24} to multi-controlled $k$-qubit unitaries. 
\end{abstract} 

\section{Introduction}

Decomposing arbitrary unitaries into smaller, structured circuits is a fundamental and well-known problem in quantum computation. 
This question is not only mathematically interesting---aiming to understand how complex unitary transformations arise from simple building blocks---but also of practical significance, as it forms the basis of quantum circuit compilation. 
There is a long history: 
early work studied exact synthesis using continuous gate sets, such as single-qubit rotations together with CNOT gates \cite{BBCD+95,MVBSM04}. 
Researchers have also investigated discrete universal gate sets as they can approximate any unitary to arbitrary precision, with the Solovay--Kitaev theorem providing the first general efficiency guarantee for single-qubit approximation \cite{DN06}.

Among discrete universal gate sets, particular attention has been given to decompositions into \emph{Clifford+T} circuits. This focus is motivated by fault-tolerant quantum computing: Clifford operations can typically be implemented at low cost, while T gates are expensive due to the overhead of magic state distillation and injection. It is therefore natural to ask \emph{how arbitrary unitaries can be compiled using as few T gates as possible, and what the optimal T-count is}. 

Although this general question remains open, progress has been made in certain special cases. 
One example is the decomposition of single-qubit unitaries into Clifford+T circuits \cite{KMM13,Sel15,RS16}, restated in \Cref{lem:single_qubit_H_T} and used as a subroutine in our work. 
In this setting, the focus has been on optimizing constant factors, which are crucial for practical implementations.

Another important case is quantum state preparation, which corresponds to implementing the first column of an $n$-qubit unitary. 
With only a constant number of ancillae, 
the T-count for quantum state preparation can be shown to have a lower bound of $\Omega(2^n/n)$\footnote{In the introduction, we treat $\epsilon$ as a constant in order to focus on the scaling in $n$. Our main result, \Cref{thm:main-thm}, is stated in full generality with explicit $\epsilon$-dependence.}
using a counting argument \cite{LKS24}.
Surprisingly, it was shown that if exponentially many ancillae are allowed, the number of T gates can be substantially smaller than the number of the Clifford gates \cite{LKS24}, and an optimal T-count of $\Theta(2^{n/2})$ was recently established in \cite{GKW24}. 

For general unitary synthesis, when only a constant number of ancillae is available, a lower bound of $\Omega(2^{2n}/n)$ can be derived in a similar manner. 
However, \cite{LKS24} showed that one can achieve a T-count of $O(2^{3n/2}\cdot n)$ 
using $O(2^{n/2})$ ancillae, 
by reducing the task to preparing $2^n$ quantum states. 
This was the first result demonstrating an asymptotic saving in T gates for unitary synthesis, and it appeared in 2018. 
Subsequently, in 2021, Rosenthal proved a query upper bound (with queries to classical Boolean functions) of $O(2^{n/2})$, which translates into a T-count of roughly the same order, $O(2^{3n/2})$, while using $O(2^n)$ ancillae \cite{Ros23}. 
For a summary of these results and other related scalings, see \cite[Table 1]{LKS24}.
Since then, no further improvements in synthesis algorithms have been obtained.
Given arbitrarily many ancillae, the current best lower bound is $\Omega(2^n)$, given in \cite{GKW24} (restated in \Cref{thm:lower_bound} for completeness). 
Determining the optimal T-count for unitary synthesis remains an open problem. 
Two natural candidate scalings have been discussed informally: 
$2^{3n/2}$, suggested by the best known upper bound from two distinct approaches, and $2^n$, motivated by the analogy with quantum state preparation, which exhibits a quadratic saving.

In this paper, we give a unitary synthesis algorithm that reduces the T-count to $O(2^{4n/3}\cdot n^{2/3})$. 
This provides further evidence that the optimal scaling could be as low as $2^n$. 

\begin{restatable}[Main result]{theorem}{mainthm}
    \label{thm:main-thm}
    Let $\epsilon > 0$ and set $L = n+\log(1/\epsilon)$. Then any $U\in \U(2^n)$ can be $\epsilon$-approximated by a Clifford+T circuit using 
    \begin{equation*}
        O\big(2^{4n/3}\cdot L^{2/3} + 2^n \cdot L\big)\quad  \text{T gates and} \quad O\big(2^{2n/3}\cdot L^{1/3} + L\big) \text{ ancillae}. 
    \end{equation*}
    In particular, for any positive integer $k \leq (n - \log_2 L)/3$, $U$ can be $\epsilon$-approximated by a Clifford+T circuit using 
    \begin{equation*}
        O\big(2^{(3n-k)/2}\cdot \sqrt{L}\big)\quad  \text{T gates and} \quad O\big(2^{(n+k)/2}\cdot \sqrt{L}\big) \text{ ancillae}. 
    \end{equation*}
\end{restatable}

The notion of approximation by a Clifford+T circuit in \Cref{thm:main-thm} is defined as follows.

\begin{definition}[Clifford+T approximation]\label{def:cliffordT_approx}
    Let $U\in \U(2^n)$. We say that \emph{$U$ admits an exact Clifford+T implementation using $\ell$ T gates and $m$ ancillae} if there exists
    $C\in \U(2^{n+m})$ such that $C$ can be written as a product of $\ell$ T gates and arbitrarily many Clifford gates, and 
    \begin{equation*}
        U \otimes \ket{0^m} = C\cdot (I_{2^n} \otimes \ket{0^m}). 
    \end{equation*}
    Given $V\in \U(2^{n+m})$ and $\epsilon\geq 0$, we say that \emph{$V$ implements $U$ to error $\epsilon$} if 
    \begin{equation*}
        \norm{U \otimes \ket{0^m} - V\cdot (I_{2^n} \otimes \ket{0^m})}\leq \epsilon. 
    \end{equation*}
    In particular, if $V$ also admits an exact Clifford+T implementation, then we say that \emph{$U$ can be $\epsilon$-approximated by a Clifford+T circuit. }
\end{definition}

\paragraph{Proof overview}

We begin by showing that any $n$-qubit unitary can be decomposed into a product of $2^n-1$ multi-controlled single-qubit unitaries, obtained via a recursive application of the cosine-sine decomposition. 
Each step of the recursion halves the dimension and introduces multi-controlled rotations, which ultimately yields $2^n-1$ such gates (\Cref{thm:recursive_CSD}). 

By carefully organizing this recursion, we observe that many of these controlled unitaries share the same target qubits and can therefore be grouped together. 
More concretely, let $k\in [n-1]$. Then the product of certain consecutive blocks of $2^k-1$ controlled unitaries has the same $k$-qubit target register (say, the first $k$ qubits), and hence simplifies to a single multi-controlled $k$-qubit unitary. 
There are in total $2^{n-k}$ such consecutive blocks. 

To implement these more general blocks, we extend the optimal algorithm for synthesizing diagonal unitaries from \cite{GKW24} to handle multi-controlled $k$-qubit unitaries, using a polynomial factoring technique. 
A diagonal unitary is a special case of a multi-controlled single-qubit unitary. 
The algorithm of \cite{GKW24} delegates the application of Hadamard and T gates to a specific target controlled by the other $n-1$ qubits to a Boolean function of $n-1$ variables. 
This idea can be made more concrete via the following example. 
Suppose $U = \sum_{x\in \{0,1\}^{n-1}}\ketbra{x} \otimes V_x$ is a multi-controlled single-qubit unitary where each $V_x\in U(2)$ can be written as a product of Hadamard and T gates: 
\begin{equation*}
    V_x = H^{f_1(x)} \cdot T^{f_2(x)} \cdots H^{f_{2L-1}(x)} \cdot T^{f_{2L}(x)},
\end{equation*}
with Boolean functions $f_i: \{0,1\}^{n-1} \to \{0,1\}$. 
If we can implement the corresponding Boolean function oracles $\ket{x}\otimes\ket{0} \mapsto \ket{x}\otimes\ket{f_i(x)}$, then $U$ can be realized by applying each Hadamard or T gate controlled by the ancilla register holding $\ket{f_i(x)}$. 

In our setting, each multi-controlled single-qubit unitary in the block can similarly be described by a Boolean function $f: \{0,1\}^n \to \{0,1\}$. 
Moreover, since their targets lie only on the first $k$ qubits, we find it convenient to factor each Boolean function $f$ as 
\begin{equation*}
    f(x_1,\ldots, x_n) = \sum_i g_i(x_1,\ldots,x_k)\cdot h_i(x_{k+1},\ldots,x_n), 
\end{equation*}
where each $g_i$ is a polynomial in the first $k$ variables and each $h_i$ is a polynomial in the remaining $n-k$ variables. 
By treating $g_i$ and $h_i$ separately, we obtain a T-count of $O(2^{\frac{n+k}{2}}\sqrt{k} + 4^k\cdot k)$ for implementing a multi-controlled $k$-qubit unitary (\Cref{lem:generalized_control_unitary}). 

By setting $k\approx n/3$, this approach yields a substantial saving over the naïve strategy of decomposing each multi-controlled $k$-qubit unitary into a product of $2^k-1$ multi-controlled single-qubit unitaries (\Cref{subsec:naive}).

\paragraph{The tradeoff between T gates and ancillae}

Let $\lambda$ denote the number of ancillae and $R$ the number of T gates. 
There is a notable tradeoff between the space (ancilla-count) and time (T-count) complexity in all related synthesis algorithms. 

For the state preparation problem, it was first observed in \cite{LKS24} and later refined by the algorithm in \cite{GKW24} that 
\begin{equation*}
    (n + \lambda)\cdot R=\Omega(2^n),
\end{equation*}
and that there exists an algorithm achieving T-count $O(2^n/\lambda)$ when $\lambda = O(2^{n/2})$. 
This tradeoff is essentially optimal, since even with arbitrarily many ancillae there is a lower bound of $\Omega(2^{n/2})$ on the T-count \cite[Theorem 4.1]{GKW24}.

For general unitary synthesis, we can derive an analogous lower bound: 
\begin{equation*}
    (n+\lambda)\cdot R=\Omega(2^{2n}).
\end{equation*}
In the regime $\lambda = O(2^{n/2})$, \cite{LKS24} gave an algorithm that uses $\lambda$ ancillae and $O(2^{2n}\cdot n/\lambda)$ T gates. 
However, the regime of this tradeoff is far from optimal, since in principle $\lambda$ could scale up to $O(2^n)$. Extrapolating this regime would then suggest the possibility of achieving a T-count as small as $O(2^n)$. 

Our algorithm essentially extends this tradeoff to the larger regime $\lambda = O(2^{2n/3}\cdot n^{1/3})$. 
This also indicates that, in order to further reduce the T-count, one would need a more sophisticated method that leverages substantially more ancillae, possibly up to $O(2^n)$.

\paragraph{Applications}

Implementing arbitrary unitaries on $n$ qubits is a common subroutine in quantum algorithms, appearing for example in first-quantized quantum simulation \cite{BREA+24,SBWRB21} and in the preparation of matrix product states \cite{HLSW25,FHZK+24,BTKW+25}. Such unitaries are typically specified by $4^n$ matrix elements stored classically, and synthesizing them with minimal T-count is crucial for practical implementations. Once fault-tolerant quantum computers are available, unitary synthesis will form part of the standard compilation toolchain for higher-level algorithmic primitives. Our synthesis technique could thus be used to reduce the cost of these primitives, where classical preprocessing can guide quantum circuit construction. Furthermore, since T gates are expected to be relatively expensive on certain hardware platforms, such as neutral-atom architectures, lowering the T-count through our method may be especially impactful.

\paragraph{Lower bound}
We include, for completeness, the current best known lower bound of $\Omega(2^n)$ on the T-count of unitary synthesis, and conjecture that it is tight when $\epsilon$ is constant. 
\begin{theorem}[{\cite[Theorem 4.3]{GKW24}}]\label{thm:lower_bound}
    There exists $U\in \U(2^n)$ such that the following is true. 
    For any integer $m\geq 0$, let $\calU$ be the associated quantum channel given by $\calU(\rho) \coloneqq U\rho U^\dagger \otimes \ketbra{0^m}$. 
    For any adaptive Clifford+T circuit $\calA$ with
    $\norm{\calA - \calU}_\diamond\leq \epsilon$, $\calA$ must use $\Omega(2^n \cdot \sqrt{\log (1/\epsilon)} + \log (1/\epsilon))$ T gates. In particular, this T count is the expectation over the randomness in the measurement outcomes in $\calA$ with worst-case input. 
\end{theorem}

\paragraph{Notations}
Throughout this paper, we use $\mathrm{i} = \sqrt{-1}$ to denote the imaginary unit, $I_N$ for the $N\times N$ identity matrix, and $[N] = \{1,2,\ldots, N\}$. 
$\U(N)$ denotes the unitary group of all $N\times N$ unitary matrices and $\SU(N)$ denotes the special unitary group of all $N\times N$ unitary matrices with determinant $1$. 
We write $\norm*{A}$ for the operator norm of a matrix $A$ and $\norm*{\calA}_\diamond$ for the diamond norm of a quantum channel $\calA$.

\section{Preliminaries}

In this section, we include a few lemmas that will be used frequently in this paper. 

The proof of the first lemma below is fairly standard via a telescoping sum argument. 
We nevertheless include it in \Cref{app:supp_proof} for completeness.

\begin{restatable}[Composition error bound]{lemma}{lemcompositionerrorbound}
\label{lem:telescope_operator_norm_ineq}
    Suppose that $V_i \in \U(2^{n+m_i})$ implements $U_i\in \U(2^n)$ to error $\epsilon$ for some integer $m_i\geq 0$. 
    Let $m = \max_{i\in [L]} m_i$. Then $(V_1\otimes I_{2^{m-m_1}})\cdots (V_L\otimes I_{2^{m-m_L}}) \in \U(2^{n+m})$ implements $U_1 \cdots U_L$ to error $L\epsilon$. 
\end{restatable}

Given a bitstring $a\in \{0,1\}^n$, denote by $\mathrm{wt}(a)$ the number of $1$'s in $a$. 
\begin{lemma}[Generating all monomials]\label{lem:all_monomials}
    The monomials generating unitary of degree $n$ given by
    \begin{equation*}
        \ket{x_1,\ldots,x_n}\otimes \bigotimes_{a\in \{0,1\}^n, \mathrm{wt}(a)\geq 2}\ket{ y_a} \mapsto \ket{x_1,\ldots,x_n} \otimes \bigotimes_{a\in \{0,1\}^n, \mathrm{wt}(a)\geq 2}\ket{ y_a \oplus x_1^{a_1}x_2^{a_2}\cdots x_n^{a_n}}, 
    \end{equation*}
    where $x_1,\ldots,x_n,y_a\in \{0,1\}$,
    can be written exactly as a product of $2^n - n - 1$ Toffoli gates. 
\end{lemma}
\begin{proof}
    We trivially have all the degree-$1$ monomials $x_1,\ldots,x_n$. 
    The algorithm then generates each of the degree-$2$ monomial using one Toffoli gate. More concretely, to generate $x_{i_1}x_{i_2}$ for some $1\leq i_1<i_2\leq n$, apply a Toffoli gate which is controlled on $\ket{x_{i_1}}$ and $\ket{x_{i_2}}$ and acts on $\ket{y_a}$ where $a \in \{0,1\}^n$ has two $1$'s at positions $i_1$ and $i_2$ and $0$'s elsewhere, i.e.\ $a=e_{i_1} + e_{i_2}$. 
    The algorithm can thus work recursively to generate all monomials. 
    For each integer $i\in [2, n]$, the algorithm can generate all degree-$i$ monomials by applying $\binom{n}{i}$ Toffoli gates controlled on the appropriate degree-$(i-1)$ monomials and degree-$1$ monomials. 
    Overall, the number of Toffoli gates used to generate all degree-$n$ monomials is $\sum_{i=2}^n \binom{n}{i}  = 2^n - n - 1$. 
\end{proof}

\begin{lemma}[T-count for Boolean function oracles, {\cite[Theorem 2]{LKS24}}]
    \label{lem:decomp_bool_oracle}
    Let $r\geq 1$ be an integer and $f:\{0,1\}^n \to \{0,1\}^r$ be an arbitrary Boolean function. Define $U_f$ as the unitary mapping $\ket{x}\ket{y}$ to $\ket{x}\ket{y\oplus f(x)}$ for all $x\in \{0,1\}^n$ and $y\in \{0,1\}^r$. Then $U_f$ admits an exact Clifford+T implementation using $O(\sqrt{r\cdot 2^n})$ T gates and ancillae. 
\end{lemma}
We also refer readers to \cite[Remark 2.2]{GKW24} for a nice proof sketch of \Cref{lem:decomp_bool_oracle}.

\begin{lemma}[Single-qubit Clifford+T approximation, \cite{RS16}]\label{lem:single_qubit_H_T}
    For any $\epsilon>0$ and any $U\in \SU(2)$, there exists $\widetilde{U}\in \SU(2)$ such that $\norm*{U-\widetilde{U}}\leq \epsilon$ and $\widetilde{U}$ is a product of $O(\log(1/\epsilon))$ Hadamard and T gates. 
\end{lemma}

\section{Cosine-sine decomposition}

In this section, we recall the cosine-sine (CS) decomposition, and explain how it can be used recursively to factorize any unitary into multi-controlled single-qubit unitaries. This will serve as the structural backbone for our synthesis results.

We begin by clarifying a piece of terminology that will be used frequently throughout the paper.
\begin{definition}[Multi-controlled unitaries]
Let $k$ be a positive integer smaller than $n$. 
We call $U\in \U(2^n)$ an \emph{$(n-k)$-fold controlled $k$-qubit unitary} if, up to a permutation of qubits,
\begin{equation*}
U = \sum_{x\in \{0,1\}^{n-k}} \ketbra{x}\otimes V_x, \qquad \text{where } V_x\in \U(2^k).
\end{equation*}
When the number of control qubits is clear from context, we simply refer to $U$ as a \emph{multi-controlled $k$-qubit unitary}.
In particular, if $U\in \U(2^n)$ is described this way, the number of control qubits is understood to be $n-k$.
\end{definition}

The form $\sum_x \ketbra{x}\otimes V_x$ is block-diagonal and naturally connects to the cosine-sine (CS) decomposition \cite{PW94}.
In this paper, we will only use a special case as summarized in \Cref{thm:CSD}. 
We refer interested readers to \cite[Section 2.1]{TT23} for detailed illustration and a proof of the more general case. 

\begin{theorem}[Special case of the CS decomposition]\label{thm:CSD}
For any $U\in \U(2^n)$, there exist $V_1,V_2,W_1,W_2\in \U(2^{n-1})$ and angles $\theta_1,\ldots,\theta_{2^{n-1}}\in [0,\pi/2]$ such that
\begin{equation}\label{eq:CSD}
        U = \underbrace{\begin{pmatrix}
            V_1 & \\
            & V_2
        \end{pmatrix}}_{\coloneqq V}\cdot \underbrace{\begin{pmatrix}
            C & S \\
            S & -C
        \end{pmatrix}}_{\coloneqq D} \cdot \underbrace{\begin{pmatrix}
            W_1 & \\
            & W_2
        \end{pmatrix}}_{\coloneqq W},
    \end{equation}
where $C = \diag(\cos\theta_1,\ldots,\cos\theta_{2^{n-1}})$ and $S = \diag(\sin\theta_1,\ldots,\sin\theta_{2^{n-1}})$.
\end{theorem}
We remark that
\begin{itemize}
    \item $V,W\in \U(2^n)$ are block-diagonal and correspond to $1$-fold controlled $(n-1)$-qubit unitaries. They are both controlled by the first qubit and act on the remaining $n-1$ qubits; 
    \item The middle block $D\in \U(2^n)$ is a multi-controlled single-qubit unitary controlled by the last $n-1$ qubits and acts on the first qubit.
\end{itemize} 
By applying the CS decomposition recursively to $V$ and $W$, one obtains a factorization of any $n$-qubit unitary into $2^n-1$ multi-controlled single-qubit unitaries. The order in which the target qubits appear has a specific combinatorial structure.
To better describe this order, we need the following notation. 

\begin{definition}[Position of the rightmost $1$]
For each $i\in [2^n-1]$, let $t_n(i)\in [n]$ denote the position of the rightmost $1$ in the binary representation of $i$ in $n$ bits, where the most significant bit has index $1$ and the least significant bit has index $n$. 
\end{definition}
For example, $t_n(1)=n$ since the binary representation of $1$ is $0^{n-1}1$, and $t_n(2^n-2)=n-1$ since the binary representation of $2^n-2$ is $1^{n-1}0$.
When $n=3$, 
\begin{equation*}
    i:1,2,3,4,5,6,7 \quad \Rightarrow \quad t_3(i):3,2,3,1,3,2,3. 
\end{equation*}

We will use the following simple proposition in the proof of \Cref{thm:recursive_CSD}. 
\begin{proposition}\label{prop:rightmost_one}
    For each $i\in [2^{n-1}-1]$, we have that $t_n(i) = t_n(i+2^{n-1}) = t_{n-1}(i) + 1\geq 2$. 
\end{proposition}

\begin{theorem}[Recursive CS decomposition]\label{thm:recursive_CSD}
For any $U\in \U(2^n)$, there exist $U_1,\ldots,U_{2^n-1}\in \U(2^n)$ such that
$U = U_1U_2\cdots U_{2^n-1}$, where each $U_i$ is a multi-controlled single-qubit unitary targeting the qubit indexed by $t_n(i)$.
\end{theorem}
\begin{proof}
We will prove the claim by induction on $n$.

The base case $n=2$ follows directly from \Cref{thm:CSD}: in \Cref{eq:CSD}, we have $U_1=V$ and $U_3=W$, both targeting at the second qubit since $t_2(1)=t_2(3)=2$, and $U_2=D$, which targets the first qubit since $t_2(2)=1$.

For the induction step, assume the claim holds for $n-1$ qubits.
Given $U\in \U(2^n)$, its CS decomposition has the form $U= VDW$ as in \Cref{eq:CSD}.
By the induction hypothesis, each $V_j,W_j\in \U(2^{n-1})$ can be decomposed into multi-controlled single-qubit unitaries: for $j=1,2$, 
\begin{equation*}
    V_j = V_{j,1} V_{j,2}\cdots V_{j,2^{n-1}-1}, \qquad
    W_j = W_{j,1} W_{j,2}\cdots W_{j,2^{n-1}-1}. 
\end{equation*}
Here each $V_{j,i},W_{j,i}$ is an $(n-1)$-qubit unitary acting on qubit $t_{n-1}(i) \in [n-1]$ where $i\in [2^{n-1}-1]$.

For each $i\in [2^{n-1}-1]$, define
\begin{equation*}
    U_i = \begin{pmatrix}
        V_{1,i} & \\
        & V_{2,i}
    \end{pmatrix}
    \quad\text{and}
    \quad
    U_{i+2^{n-1}} = \begin{pmatrix}
        W_{1,i} & \\
        & W_{2,i}
    \end{pmatrix}. 
\end{equation*}
Each $U_i$ (or $U_{i+2^{n-1}}$) inherits the control structure from $V_{j,i}$ (or $W_{j,i}$), with the first qubit acting as an additional control. 
Then by \Cref{prop:rightmost_one}, each $U_i$ targets the qubit indexed by $t_{n-1}(i) + 1 = t_n(i)$ and is controlled by the remaining $n-1$ qubits. 
Similarly, each $U_{i + 2^{n-1}}$ targets the qubit indexed by $t_{n-1}(i) + 1 = t_n(i + 2^{n-1})$ and is controlled by the rest of the qubits. 

Finally, let us set $U_{2^{n-1}}=D$, which is a multi-controlled single-qubit unitary acting on the first qubit and indeed $t_n(2^{n-1}) = 1$.
This completes the induction.
\end{proof}

\subsection{A naïve implementation using \texorpdfstring{$O(2^{3n/2} \cdot n^{1/2})$}{O(2³ⁿᐟ² · n¹ᐟ²)} T gates}\label{subsec:naive}

It was shown in \cite[Theorem 1.2]{GKW24} that any diagonal unitary $D\in \U(2^n)$ can be $\epsilon$-approximated by a Clifford+T circuit using $O(\sqrt{2^n\cdot \log(1/\epsilon)} + \log(1/\epsilon))$ T gates and ancillae. 
In fact, if one takes a closer look at the proof of \cite[Theorem 1.2]{GKW24}, it directly works for any unitary in the form of
\begin{equation}\label{eq:multi_control_SU2}
    \sum_{x\in \{0,1\}^{n-1}} \ketbra{x} \otimes R_x, \quad \text{as long as}\quad R_x\in \SU(2). 
\end{equation}
Hence, one can relax the theorem to any multi-controlled single-qubit unitary $U$, not just the diagonal ones.
This is because $U$ can be written as a product of a diagonal gate and a gate in the form of \Cref{eq:multi_control_SU2}. 
We summarize this generalization as below. 
\begin{corollary}[T-count for multi-controlled single-qubit unitaries]\label{thm:GKW_T_count}
    For any $\epsilon>0$ and any multi-controlled single-qubit unitary $U\in \U(2^n)$, there exists $\widetilde{U}\in \U(2^{n+1})$ such that $\widetilde{U}$ implements $U$ to error $\epsilon$ and $\widetilde{U}$ admits an exact Clifford+T implementation using $O(\sqrt{2^n \cdot \log(1/\epsilon)} + \log(1/\epsilon))$ T gates and ancillae. 
\end{corollary}
\begin{proof}
    Any multi-controlled single-qubit unitary $U\in \U(2^n)$, up to a permutation of qubits, can be written as,
    \begin{align*}
        U &= \sum_{x\in \{0,1\}^{n-1}} \ketbra{x} \otimes V_x  \tag{each $V_x\in \U(2)$}\\
        &= \sum_{x\in \{0,1\}^{n-1}} \ketbra{x} \otimes (e^{\ii \theta_x} \cdot R_x) \tag{each $R_x\in \SU(2)$, $\theta_x\in [0, 2\pi)$}\\
        &= \Big(\underbrace{\sum_{x\in \{0,1\}^{n-1}} \ketbra{x} \otimes (e^{\ii \theta_x}\cdot I_2)}_{\coloneqq D} \Big)\cdot \Big(\underbrace{\sum_{x\in \{0,1\}^{n-1}} \ketbra{x} \otimes R_x}_{\coloneqq R}\Big). 
    \end{align*}
    It follows from the proof of \cite[Theorem 1.2]{GKW24} that 
    \begin{itemize}
        \item there exists $\widetilde{R} \in \SU(2^n)$ such that $ \widetilde{R}$ implements $R$ to error $\epsilon/2$, i.e.\
        $\norm*{R - \widetilde{R}}\leq \epsilon/2$.  
        \item there exists $\widetilde{D}\in \SU(2^{n+1})$ such that $\widetilde{D}$ implements $D$ to error $\epsilon/2$. 
        \item Both $\widetilde{R}$ and $\widetilde{D}$ admit exact Clifford+T implementations using $O(\sqrt{2^n \cdot \log(1/\epsilon)} + \log(1/\epsilon))$ T gates and ancillae.
    \end{itemize}
    By \Cref{lem:telescope_operator_norm_ineq}, 
    we know that $\widetilde{U}=\widetilde{D}\cdot (\widetilde{R} \otimes I_2)$ implements $U = DR$ to error $\epsilon$, which completes the proof. 
\end{proof}

Following from \Cref{thm:recursive_CSD}, we know that any $U\in \U(2^n)$ can be written as a product of $2^n-1$ multi-controlled single-qubit unitaries, i.e.\ $U = U_1U_2\cdots U_{2^n-1}$. 
By \Cref{thm:GKW_T_count}, for each $i\in [2^n-1]$, there exists $\widetilde{U}_i \in \U(2^{n+1})$ such that $\widetilde{U}_i$ implements $U_i$ to error $\epsilon \cdot 2^{-n}$ and $\widetilde{U}_i$ admits an exact Clifford+T implementation using 
\begin{equation*}
    O(\sqrt{2^{n} \cdot \log(1/(\epsilon \cdot 2^{-n}))} + \log(1/(\epsilon \cdot 2^{-n}))) = O\!\left(\sqrt{2^{n}(n+ \log(1/\epsilon))} + \log(1/\epsilon) \right)
\end{equation*}
T gates and ancillae. 
So naively, let $\widetilde{U} = \widetilde{U}_1 \widetilde{U}_1 \cdots \widetilde{U}_{2^n - 1} \in \U(2^{n+1})$. 
Then by \Cref{lem:telescope_operator_norm_ineq}, $\widetilde{U}$ implements $U$ to error $\epsilon$ and $\widetilde{U}$ admits an exact Clifford+T implementation with
\begin{equation*}
    \text{T-count: }
    O\!\left(2^n \left(  \sqrt{2^{n}(n+ \log(1/\epsilon))} + \log(1/\epsilon) \right)\right), \qquad \text{ancilla-count}: O\!\left(\sqrt{2^{n}(n+ \log(1/\epsilon))} + \log(1/\epsilon) \right). 
\end{equation*}
Thus, our naïve recursive CS-based synthesis achieves a T-count of $O(2^{3n/2}\cdot n^{1/2})$, improving slightly over the previous best $O(2^{3n/2}\cdot n)$ scaling due to \cite{LKS24} while relying on an arguably simpler analysis.

\section{Lower the T-count}

The key observation that lowers the T-count is that many of the controlled unitaries in the recursive CS decomposition can be grouped together. 

Throughout this section, we will write $t_n(i)$ simply as $t(i)$, since the subscript is always $n$. 
Let $k \in [n-1]$ be a parameter that we will optimize later. 
To begin with, note that each of $U_1,\ldots, U_{2^k-1}$ acts on one of the last $k$ qubits, since $t(1),\ldots, t(2^k-1)\in \{n-k+1, \ldots, n\}$. 
Their product $U_1\cdots U_{2^k-1}$ is thus a unitary acting on qubits in $[n]\setminus [n-k]$ and controlled by qubits in $[n-k]$, i.e.\
\begin{equation*}
    W_0 \coloneqq U_1\cdots U_{2^k-1} = \sum_{x\in \{0,1\}^{n-k} } \ketbra{x} \otimes V_{0,x}, \qquad \text{where }V_{0,x}\in \U(2^k). 
\end{equation*}
Similarly, since $t(2^k + i) = t(i)$ for any $i\in [2^k-1]$, each of $U_{2^{k} + 1}, \ldots, U_{2^{k} + 2^k - 1}$ also acts on one of the last $k$ qubits. 
So their product $W_1 \coloneqq U_{2^{k} + 1}\cdots U_{2^{k} + 2^k - 1}$ is also a multi-controlled $k$-qubit unitary. 
Let us fully generalize this. 
For each $j\in \{0, 1,\ldots, 2^{n-k}-1\}$, let
\begin{equation}\label{eq:product_W}
    W_j \coloneqq U_{j\cdot 2^{k} + 1}\cdot U_{j\cdot 2^{k} + 2} \cdots U_{j\cdot 2^{k} + 2^{k}-1} = \sum_{x\in \{0,1\}^{n-k}} \ketbra{x} \otimes V_{j,x},\qquad \text{where }V_{j,x}\in \U(2^k).
\end{equation}
Overall,
\begin{equation}\label{eq:decompose_U_into_Ws}
    U = W_0 \cdot \prod_{j=1}^{2^{n-k}-1} U_{j\cdot 2^k}\cdot W_{j}. 
\end{equation}

If we implement each $W_j$, a product of $2^k-1$ controlled unitaries, to error $\epsilon \cdot 2^{-(n-k)}$ naively as in \Cref{subsec:naive}, the T-count is $O(2^k (  \sqrt{2^{n}(n+ \log(1/\epsilon))} + \log(1/\epsilon) ))$. 
If we ignore the polynomial dependence on $n$ and take $\epsilon$ to be constant, this T-count is on the order of $2^{k + n/2}$. 
However, we will show in \Cref{lem:generalized_control_unitary} that there is a cheaper way to implement $W_j$, with cost roughly on the order of $2^{(n+k)/2} + 2^{2k}$. 

\paragraph{Notation}
Let $U_i$ be a multi-controlled single-qubit unitary targeting qubit $t(i)$, i.e.\
\begin{equation*}
    U_i = \sum_{x\in \{0,1\}^{n-1}} \ketbra{x_1,\ldots, x_{t(i)-1}} \otimes V_{i,x} \otimes \ketbra{x_{t(i)},\ldots, x_{n-1}},
\end{equation*}
where $V_{i,x}\in \U(2)$. 
We adopt the following notation to simplify the writing of $U_i$ as
\begin{equation*}
    U_i = \sum_{x\in \{0,1\}^{n-1}} \ketbra{x} \otimes \left[V_{i,x}\right]_{t(i)},
\end{equation*}
where the subscript $t(i)$ in $\left[A\right]_{t(i)}$ indicates that $A$ actually acts on the qubit indexed by $t(i)$ and the control qubits are indexed by $[n]\setminus \{t(i)\}$.

\subsection{Implementing multi-controlled \texorpdfstring{$k$}{k}-qubit unitaries}

In \Cref{lem:special_k_control}, we first show a special case of implementing a product of $m$ multi-controlled single-qubit unitaries, 
each acting on one of $k$ designated target qubits.
Then we plug in $m = O(2^k)$ and bound the cost of implementing a general multi-controlled $k$-qubit unitary in \Cref{lem:generalized_control_unitary}. 

\begin{lemma}[T-count for a product of multi-controlled single-qubit unitaries]\label{lem:special_k_control}
    Let $n,m,k$ be positive integers and $k<n$.  
    For each $i\in [m]$, let 
    \begin{equation*}
        U_i = \sum_{x\in \{0,1\}^{n-1}}\ketbra{x} \otimes \left[R_{i, x}\right]_{h_i}, \qquad \text{where }R_{i,x}\in \SU(2) \text{ and } h_i\in [k]. 
    \end{equation*}
    Let $U = U_1U_2\cdots U_m$. Then for any $\epsilon>0$, there exists $\widetilde{U}\in \U(2^n)$ such that $\norm*{U - \widetilde{U}}\leq \epsilon$ and $\widetilde{U}$ admits an exact Clifford+T implementation using $O(2^{n/2}\cdot \sqrt{m\cdot \log(m/\epsilon)} + 2^k\cdot m\cdot \log(m/\epsilon))$ T gates and ancillae.
\end{lemma}
\begin{proof}
    We first describe the construction of $\widetilde{U} \in \U(2^n)$. 
    For each $i\in [m]$ and $x\in \{0,1\}^{n-1}$, since $R_{i, x} \in \SU(2)$, it follows from \Cref{lem:single_qubit_H_T} that there exists $\widetilde{R}_{i,x}\in \SU(2)$ such that $\norm*{R_{i,x} - \widetilde{R}_{i,x}} \leq \epsilon/m$ and $\widetilde{R}_{i,x}$ can be written exactly as a product of $L$ Hadamard and $L$ T gates where $L = \lceil c\cdot \log (m/\epsilon) \rceil$ for some constant $c>0$. 
    More concretely, let us write
    \begin{equation*}
        \widetilde{R}_{i,x} = H^{f_{i,1}(x)} \cdot T^{f_{i,2}(x)} \cdots H^{f_{i,2L-1}(x)} \cdot T^{f_{i,2L}(x)}, 
    \end{equation*}
    where $f_{i, j}: \{0,1\}^{n-1} \mapsto \{0,1\}$ for each $j\in [2L]$.  
    Then $\widetilde{U}_i \coloneqq \sum_{x\in \{0,1\}^{n-1}}\ketbra{x} \otimes [\widetilde{R}_{i, x}]_{h_i}$ implements $U_i$ to error $\epsilon/m$. 
    Overall, by \Cref{lem:telescope_operator_norm_ineq}, 
    $\widetilde{U} \coloneqq \widetilde{U}_1\widetilde{U}_2\cdots \widetilde{U}_m$ implements $U$ to error $\epsilon$.  

    It remains to calculate the T-count and ancilla-count for implementing $\widetilde{U}$. 

    \begin{remark}
    Let us take a detour and recall the high-level idea for implementing each $\widetilde{U}_i$. 
    The Boolean function $f_{i,j}$ can be written as a polynomial in $n-1$ variables $x_1, \ldots, x_{h_i-1}, x_{h_i+1},\ldots, x_{n}$ over $\F_2$. 
    Let us write $x = (x_1, \ldots, x_{h_i-1}, x_{h_i+1},\ldots, x_{n}) \in \{0,1\}^{n-1}$. 
    Suppose that we can implement 
    \begin{equation}\label{eq:encode_control_logic}
        \ket{x}\otimes \left(\bigotimes_{j\in [2L]} \ket{y_j}\right) \mapsto
        \ket{x} \otimes \left(\bigotimes_{j\in [2L]} \ket{y_j \oplus f_{i,j}(x)} \right).
    \end{equation}
    Then we can implement $\widetilde{U}_i$ as follows: 
    \begin{enumerate}
        \item prepare $2L$ ancillae in $\ket{0^{2L}}$. 
        \item run \Cref{eq:encode_control_logic} once on all qubits except the one indexed by $h_i$ (whose state is $\ket{x_{h_i}}$). 
        \item sequentially apply $L$ controlled-Hadamard and $L$ controlled-T gates, each of which is controlled on the corresponding qubit with state $\ket{f_{i,j}(x)}$ and acts on the same qubit indexed by $h_i$. 
        \item run \Cref{eq:encode_control_logic} once to uncompute each ancilla register with state $\ket{f_{i,j}(x)}$ back to $\ket{0}$. 
    \end{enumerate}
    In summary, 
    \begin{align*}
        \ket{x} \otimes \ket{x_{h_i}} \otimes \ket{0^{2L}} &\mapsto \ket{x} \otimes \ket{x_{h_i}} \otimes \left(\bigotimes_{j\in [2L]} \ket{f_{i,j}(x)} \right) \tag{step $2$}\\
        &\mapsto \ket{x}\otimes \left(\widetilde{R}_{i,x}\cdot \ket{x_{h_i}}\right) \otimes \left(\bigotimes_{j\in [2L]} \ket{f_{i,j}(x)} \right) \tag{step 3}\\
        &\mapsto \ket{x}\otimes \left(\widetilde{R}_{i,x}\cdot \ket{x_{h_i}}\right) \otimes \ket{0^{2L}}  \tag{step 4}.
    \end{align*}
    So to implement $\widetilde{U}$, it boils down to the Clifford+T implementation of 
    \Cref{eq:encode_control_logic} for each $i\in [m]$ and how their product can be optimized altogether. 
    \end{remark}

    We now continue the proof of \Cref{lem:special_k_control}. 
    From now on we write $x=(x_1,\ldots,x_n)$, and each $f_{i,j}$ is extended to a polynomial in $n$ variables (i.e.\ its dependence on $x_{h_i}$ is trivial). 

    Since all the targeting qubits $h_i$ can only be in the first $k$ qubits, we can decompose $f_{i,j}$ in a way that separates the variables $x_1,\ldots,x_{k}$ from the rest: for each $i\in [m]$ and $j\in [2L]$, 
    \begin{equation}\label{eq:poly_factor}
        f_{i,j}(x) = \sum_{b\in \{0,1\}^{k}} (x_1^{b_1} \cdots x_{k}^{b_k}) \cdot g_{i,j,b}(x_{k+1}, \ldots, x_n),
    \end{equation}
    where $g_{i,j,b}$ is a polynomial in $x_{k+1}, \ldots, x_n$ over $\F_2$. 
    We denote all the $g_{i,j,b}$ polynomials together by 
    \begin{equation}\label{eq:all_polys_g}
        g:\{0,1\}^{n-k} \to \{0,1\}^{m\cdot(2L)\cdot 2^{k}}, \qquad \text{where} \quad g(x_{k+1},\ldots, x_n)_{i,j,b} = g_{i,j,b}(x_{k+1},\ldots, x_n). 
    \end{equation}

    Now we are ready to describe the final algorithm: 
    \begin{enumerate}
        \item Prepare $M$ ancillae in $\ket{0^M}$. Denote the input state as $\ket{x_1,\ldots,x_{n}} \otimes \ket{0^M}$. We will specify $M$ later. 

        \item Recall $g:\{0,1\}^{n-k} \to \{0,1\}^{m\cdot (2L)\cdot 2^{k}}$ defined in \Cref{eq:all_polys_g}. 
        Using \Cref{lem:decomp_bool_oracle} with $r = m\cdot (2L)\cdot 2^{k}$, the unitary $U_g$ given by
        \begin{equation*}
            \ket{x_{k+1},\ldots, x_n}\otimes \Big(\bigotimes_{\substack{i\in[m],j\in [2L],\\b\in\{0,1\}^{k}}}\ket{y_{i,j,b}}\Big) \mapsto \ket{x_{k+1},\ldots, x_n}\otimes \Big(\bigotimes_{\substack{i\in[m],j\in [2L],\\b\in\{0,1\}^{k}}}\ket{y_{i,j,b} \oplus g_{i,j,b}(x_{k+1},\ldots, x_n)}\Big)
        \end{equation*}
        admits an exact Clifford+T implementation using $O(\sqrt{2^{n-k}\cdot r}) = O(2^{n/2}\sqrt{mL})$ T gates and ancillae. Applying $U_g$ gives
        \begin{equation*}
            \ket{x_1,\ldots,x_n}\otimes \Big(\underbrace{\bigotimes_{\substack{i\in[m],j\in [2L],\\b\in\{0,1\}^{k}}}\ket{g_{i,j,b}(x_{k+1},\ldots, x_n)}}_{\coloneqq \ket{A}}\Big) \otimes \ket{0^{M_A}}.
        \end{equation*}
        Here, the ancilla count must satisfy $M = \Omega(mL2^k + M_A)$ and $M_A= \Omega(2^{n/2}\sqrt{mL})$.

        \item For each $i\in [m]$:
        \begin{enumerate}
            \item Apply the monomials generating unitary in \Cref{lem:all_monomials} to produce the state that encodes all monomials in $x_1, \ldots, x_{h_i-1}, x_{h_i +1}, \ldots, x_k$. 
            With a permutation of registers, the resulting state is given by
            \begin{equation*}
            \ket{x_{h_i}} \otimes \ket{A} \otimes  \Big(\underbrace{\bigotimes_{a\in \{0,1\}^{k-1}} \ket{x_{1}^{a_1}\cdots x_{h_i-1}^{a_{h_i - 1}} x_{h_i+1}^{a_{h_i}}\cdots x_{k}^{a_{k - 1}} } }_{\coloneqq \ket{B}}\Big)\otimes \ket{0^{M_B}},
            \end{equation*}
            using $O(2^{k})$ Toffoli gates and ancillae. 
            Here, the ancilla counts must satisfy $M_A = \Omega(2^k + M_B)$. 
             
            \item Set $M_B = 2L\cdot 2^k$. By \Cref{eq:poly_factor}, one can use $2L \cdot 2^k$ Toffoli gates to produce
            \begin{equation*}
                \ket{x_{h_i}} \otimes \ket{A} \otimes \ket{B}\otimes \Big( \bigotimes_{j\in [2L]}\ket{f_{i,j}(x)}\Big). 
            \end{equation*}
    
            \item Use $L$ controlled-Hadamard and $L$ controlled-T gates to produce
            \begin{equation*}
                \left(\widetilde{R}_{i,x}\cdot \ket{x_{h_i}}\right) \otimes \ket{A} \otimes \ket{B}\otimes \Big( \bigotimes_{j\in [2L]}\ket{f_{i,j}(x)}\Big). 
            \end{equation*}
            
            \item Uncompute the state to 
            \begin{equation*}
                \ket{x_1,\ldots x_{h_i-1}} \otimes \left(\widetilde{R}_{i,x}\cdot \ket{x_{h_i}}\right) \otimes \ket{x_{h_i+1},\ldots,x_k} \otimes \ket{A} \otimes \ket{0^{M_A}}.
            \end{equation*}
        \end{enumerate}
    \end{enumerate}
    Overall, both the number of ancillae and T-count are $O(2^{n/2}\cdot \sqrt{mL} + 2^{k}\cdot mL)$. 
\end{proof}

We now bound the T-count for implementing a general multi-controlled $k$-qubit unitary. 

\begin{lemma}[T-count for multi-controlled $k$-qubit unitaries]\label{lem:generalized_control_unitary}
    Let $k$ be a positive integer smaller than $n$. 
    For any $\epsilon >0$ and any multi-controlled $k$-qubit unitary $W\in \U(2^n)$ targeting the last $k$ qubits, i.e.\
    \begin{equation*}
        W = \sum_{x\in \{0,1\}^{n-k}} \ketbra{x} \otimes V_x, \qquad \text{where } V_x\in \U(2^k), 
    \end{equation*}
    there exists $\widetilde{W} \in \U(2^{n+1})$ such that $\widetilde{W}$ implements $W$ to error $\epsilon$ and $\widetilde{W}$ admits an exact Clifford+T implementation using $O(2^{(n+k)/2} \sqrt{k+\log(1/\epsilon)} + 4^{k} (k+\log(1/\epsilon)))$ T gates and ancillae. 
\end{lemma}
\begin{proof}
    It follows from the recursive CS decomposition in \Cref{thm:recursive_CSD} (see also \Cref{eq:product_W}) that $W$ can be written as
    \begin{equation*}
        W = U_1 U_2 \cdots U_{2^k - 1},
    \end{equation*}
    where for each $i \in [2^k-1]$, $U_i$ is a multi-controlled single-qubit unitary targeting qubit $t(i)\in [n]\setminus [n-k]$, i.e.\
    \begin{equation*}
        U_i = \sum_{x\in \{0,1\}^{n-1}} \ketbra{x} \otimes \left[e^{\ii \cdot \theta_{i,x}}\cdot R_{i,x}\right]_{t(i)}, 
    \end{equation*}
    where $\theta_{i,x} \in [0, 2\pi)$, $R_{i,x} \in \SU(2)$. Here we single out the phases $e^{\ii \cdot \theta_{i,x}}$ so that each $R_{i,x}$ admits a Hadamard+T approximation following from \Cref{lem:single_qubit_H_T}. 
    Let us write
    \begin{equation*}
        R_i \coloneqq \sum_{x\in \{0,1\}^{n-1}} \ketbra{x} \otimes \left[ R_{i,x}\right]_{t(i)} \quad \text{and} \quad \Phi_i \coloneqq \sum_{x\in \{0,1\}^{n-1}} \ketbra{x} \otimes \left[\diag(e^{\ii \cdot \theta_{i,x}}, e^{\ii \cdot \theta_{i,x}})\right]_{t(i)}. 
    \end{equation*}
    So $U_i = R_i \cdot \Phi_i$. 
    To relate $\Phi_i$ to a unitary operator with determinant $1$, we use a similar trick as in the proof of \cite[Theorem 1.2]{GKW24} and consider
    \begin{equation*}
        D_i \coloneqq \Phi_i \otimes \ketbra{0}_\anc + \Phi_i^\dagger \otimes \ketbra{1}_\anc = \sum_{x\in \{0,1\}^{n-1}} \ketbra{x} \otimes [I_2]_{t(i)} \otimes \underbrace{\diag(e^{\ii \cdot \theta_{i,x}}, e^{-\ii \cdot \theta_{i,x}})}_{\coloneqq D_{i,x} \in \SU(2)}{}_\anc,
    \end{equation*}
    where the last equality is because
    \begin{align*}
        \diag(e^{\ii \cdot \theta_{i,x}}, e^{\ii \cdot \theta_{i,x}}) \otimes \ketbra{0} + \diag(e^{-\ii \cdot \theta_{i,x}}, e^{-\ii \cdot \theta_{i,x}}) \ketbra{1} &= \diag (e^{\ii \cdot \theta_{i,x}}, e^{-\ii \cdot \theta_{i,x}}, e^{\ii \cdot \theta_{i,x}}, e^{-\ii \cdot \theta_{i,x}}) \\
        &= I_2 \otimes \diag(e^{\ii \cdot \theta_{i,x}}, e^{-\ii \cdot \theta_{i,x}}). 
    \end{align*}
    This additional ancilla register $\anc$ does not affect the overall implementation of $U_i$ because 
    for any $\ket{\psi}\in \C^{2^n}$, 
    \begin{equation*}
        (U_i \ket{\psi})\otimes \ket{0} = (R_i\cdot \Phi_i\cdot \ket{\psi}) \otimes \ket{0} = (R_i \otimes I_2) \cdot D_i \cdot (\ket{\psi} \otimes \ket{0}). 
    \end{equation*}
    In other words, $(R_i \otimes I_2) \cdot D_i$ implements $U_i$ to error $0$. 

    Now, $R_i \otimes I_2$ and $D_i \in \SU(2^{n+1})$ are $n$-fold controlled single-qubit unitaries, targeting qubit $t(i) \in [n]\setminus [n-k]$ and qubit $n+1$ (i.e.\ register $\anc$) respectively. Hence, we can further simplify to write them as
    \begin{equation*}
        R_i \otimes I_2 = \sum_{x\in \{0,1\}^n } \ketbra{x} \otimes \left[R_{i,x}\right]_{t(i)}\quad \text{and}\quad D_i = \sum_{x\in \{0,1\}^n } \ketbra{x} \otimes D_{i, x},
    \end{equation*}
    where $R_{i,x}, D_{i,x}\in \SU(2)$. 
    Overall,
    \begin{align*}
        (W \ket{\psi})\otimes \ket{0} &= (U_1U_2\cdots U_{2^k-1} \ket{\psi})\otimes \ket{0} \\
        &= \left( \prod_{i=1}^{2^k-1} \left(\sum_{x\in \{0,1\}^n } \ketbra{x} \otimes \left[R_{i,x}\right]_{t(i)} \right) \cdot \left(\sum_{x\in \{0,1\}^n } \ketbra{x} \otimes D_{i,x} \right)\right) \cdot (\ket{\psi} \otimes \ket{0}). 
    \end{align*}
    Applying \Cref{lem:special_k_control} with $m = 2\cdot (2^k-1)$, we know that there exists $\widetilde{W}\in \U(2^{n+1})$ such that $\widetilde{W}$ implements $W$ to error $\epsilon$ and $\widetilde{W}$ admits an exact Clifford+T implementation using 
    $O(2^{(n+k)/2} \sqrt{k+\log(1/\epsilon)} + 4^{k} (k+\log(1/\epsilon)))$ T gates and ancillae. 
\end{proof}

\subsection{Proof of \texorpdfstring{\Cref{thm:main-thm}}{Theorem 1.1}}

\mainthm*
\begin{proof}
    It follows from the recursive CS decomposition in \Cref{thm:recursive_CSD} (see also \Cref{eq:decompose_U_into_Ws}) that
    \begin{equation*}
        U = W_0 \cdot \prod_{j=1}^{2^{n-k}-1} U_{j\cdot 2^k}\cdot W_{j}, 
    \end{equation*}
    where each $W_j$ is a multi-controlled $k$-qubit unitary acting on the last $k$ qubits. 
    We have already calculated the cost for implementing each $U_{j\cdot 2^k}$ and $W_j$: 
    \begin{itemize}
        \item By \Cref{thm:GKW_T_count}, for each $j\in \{1, \ldots, 2^{n-k}-1\}$, there exists $\widetilde{U}_{j\cdot 2^k}\in \U(2^{n+1})$ such that  $\widetilde{U}_{j\cdot 2^k}$ implements $U_{j\cdot 2^k}$ to error $\delta$ and $\widetilde{U}_{j\cdot 2^k}$ admits an exact Clifford+T implementation using $O(\sqrt{2^n\cdot \log(1/\delta)} + \log(1/\delta))$ T gates and ancillae. 

        \item By \Cref{lem:generalized_control_unitary}, for each $j\in \{0, 1,\ldots, 2^{n-k}-1\}$, there exists $\widetilde{W}_j$ such that $\widetilde{W}_j$ implements $W_j$ to error $\delta$ and $\widetilde{W}_j$ admits an exact Clifford+T implementation using $O(2^{(n+k)/2} \sqrt{k+\log(1/\delta)} + 4^{k} (k+\log(1/\delta)))$ T gates and ancillae. 
    \end{itemize}
    Let $\delta = \epsilon \cdot 2^{-(n-k)}/2$ and $L = k + \log(1/\delta) = n + \log(1/\epsilon)$. 
    Then the total ancilla-count is
    \begin{equation}\label{eq:overall_ancilla_count}
        O\!\left( \sqrt{2^n \cdot \log(1/\delta)} + \log(1/\delta) + 2^{\frac{n+k}{2}} \sqrt{k+\log(1/\delta)} + 4^{k} (k+\log(1/\delta))\right) = O\!\left( 2^{\frac{n+k}{2}} \sqrt{L} + 4^k L\right),
    \end{equation}
    and hence the total T-count is 
    \begin{equation}\label{eq:overall_T_count}
        O\!\left(2^{n-k} \cdot \left(2^{\frac{n+k}{2}} \sqrt{L} + 4^k L\right)\right) . 
    \end{equation}
    For any positive integer $k$ satisfying $2^{\frac{n+k}{2}} \sqrt{L} \geq 4^k L$, i.e.\ $2^k \leq 2^{n/3}\cdot L^{-1/3}$ or $k\leq (n - \log_2 L)/2$, we have that 
    \begin{equation}\label{eq:resource_counts_k}
        \eqref{eq:overall_ancilla_count}= O\!\left(2^{\frac{n+k}{2}}\sqrt{L}\right)  \qquad \text{and} \qquad
        \eqref{eq:overall_T_count} = O\!\left(2^{\frac{3n-k}{2}}\sqrt{L}\right)  . 
    \end{equation}
    The ancillae and T gates tradeoff follows from the fact that $2^{\frac{n+k}{2}}\cdot 2^{\frac{3n-k}{2}} = 2^{2n}$. 
    This proves the second part of \Cref{thm:main-thm}. 
    To prove the first part, we distinguish two cases and choose $k$ to minimize the T-count in each. 
    \begin{itemize}
        \item When $L\geq 2^n$, we have that $4^kL\geq 2^{\frac{n+k}{2}} \sqrt{L}$ for any $k$. Then
        \begin{align*}
            \eqref{eq:overall_ancilla_count} = O\!\left(4^k \cdot L\right) \qquad \text{and}\qquad
            \eqref{eq:overall_T_count} = O\!\left(2^{n-k} \cdot 4^k \cdot L\right) = O\!\left(2^{n+k} \cdot L\right) . 
        \end{align*}
        So we should set $k$ to be the smallest possible value, i.e.\ $k=1$. 

        \item When $L\leq 2^n$, let us set $k = \lfloor (n - \log_2 L)/3 \rfloor$. Then by \Cref{eq:resource_counts_k}, we have
        \begin{equation*}
            \eqref{eq:overall_ancilla_count} =  
            O\!\left(2^{2n/3}\cdot L^{1/3}\right) \qquad \text{and} \qquad
            \eqref{eq:overall_T_count} =  
            O\!\left(2^{4n/3} \cdot L^{2/3}\right) . 
        \end{equation*}
    \end{itemize}
    Combining the above two cases gives 
    \begin{equation*}
        \text{ancilla-count}: O\!\left(L + 2^{2n/3}\cdot L^{1/3} \right), \qquad \text{T-count}:  O\!\left(2^n \cdot L +  2^{4n/3} \cdot L^{2/3}\right).  \qedhere
    \end{equation*}
\end{proof}

\section*{Acknowledgments}

The author would like to thank Bill Huggins and Robin Kothari for many insightful discussions throughout the course of this project, Kewen Wu for comments on an early-stage note that developed into this paper, and Aram Harrow, Nathan Wiebe, and John Wright for helpful discussions. 
This work was done when the author was a Student Researcher at Google.

\bibliographystyle{alpha}
\bibliography{ref}

\newcommand{\etalchar}[1]{$^{#1}$}
\begin{thebibliography}{SBW{\etalchar{+}}21}

\bibitem[BBC{\etalchar{+}}95]{BBCD+95}
Adriano Barenco, Charles~H. Bennett, Richard Cleve, David~P. DiVincenzo, Norman Margolus, Peter Shor, Tycho Sleator, John~A. Smolin, and Harald Weinfurter.
\newblock Elementary gates for quantum computation.
\newblock {\em Phys. Rev. A}, 52:3457--3467, Nov 1995.

\bibitem[BRE{\etalchar{+}}24]{BREA+24}
Dominic~W. Berry, Nicholas~C. Rubin, Ahmed~O. Elnabawy, Gabriele Ahlers, A.~Eugene DePrince, Joonho Lee, Christian Gogolin, and Ryan Babbush.
\newblock Quantum simulation of realistic materials in first quantization using non-local pseudopotentials.
\newblock {\em npj Quantum Information}, 10(1):130, 2024.

\bibitem[BTK{\etalchar{+}}25]{BTKW+25}
Dominic~W. Berry, Yu~Tong, Tanuj Khattar, Alec White, Tae~In Kim, Guang~Hao Low, Sergio Boixo, Zhiyan Ding, Lin Lin, Seunghoon Lee, Garnet Kin-Lic Chan, Ryan Babbush, and Nicholas~C. Rubin.
\newblock Rapid initial-state preparation for the quantum simulation of strongly correlated molecules.
\newblock {\em PRX Quantum}, 6:020327, May 2025.

\bibitem[DN06]{DN06}
Christopher~M. Dawson and Michael~A. Nielsen.
\newblock The solovay-kitaev algorithm.
\newblock {\em Quantum Info. Comput.}, 6(1):81–95, January 2006.

\bibitem[FHZ{\etalchar{+}}24]{FHZK+24}
Stepan Fomichev, Kasra Hejazi, Modjtaba~Shokrian Zini, Matthew Kiser, Joana Fraxanet, Pablo Antonio~Moreno Casares, Alain Delgado, Joonsuk Huh, Arne-Christian Voigt, Jonathan~E. Mueller, and Juan~Miguel Arrazola.
\newblock Initial state preparation for quantum chemistry on quantum computers.
\newblock {\em PRX Quantum}, 5:040339, Dec 2024.

\bibitem[GKW24]{GKW24}
David Gosset, Robin Kothari, and Kewen Wu.
\newblock Quantum state preparation with optimal t-count, 2024.

\bibitem[HLSW25]{HLSW25}
William~J. Huggins, Oskar Leimkuhler, Torin~F. Stetina, and K.~Birgitta Whaley.
\newblock Efficient state preparation for the quantum simulation of molecules in first quantization.
\newblock {\em PRX Quantum}, 6:020319, Apr 2025.

\bibitem[KMM13]{KMM13}
Vadym Kliuchnikov, Dmitri Maslov, and Michele Mosca.
\newblock Fast and efficient exact synthesis of single-qubit unitaries generated by clifford and t gates.
\newblock {\em Quantum Info. Comput.}, 13(7–8):607–630, July 2013.

\bibitem[LKS24]{LKS24}
Guang~Hao Low, Vadym Kliuchnikov, and Luke Schaeffer.
\newblock Trading t gates for dirty qubits in state preparation and unitary synthesis.
\newblock {\em Quantum}, 8:1375, June 2024.

\bibitem[MVBS04]{MVBSM04}
Mikko Möttönen, Juha~J. Vartiainen, Ville Bergholm, and Martti~M. Salomaa.
\newblock Quantum circuits for general multiqubit gates.
\newblock {\em Physical Review Letters}, 93(13), September 2004.

\bibitem[PW94]{PW94}
C.C. Paige and M.~Wei.
\newblock History and generality of the cs decomposition.
\newblock {\em Linear Algebra and its Applications}, 208-209:303--326, 1994.

\bibitem[Ros23]{Ros23}
Gregory Rosenthal.
\newblock Query and depth upper bounds for quantum unitaries via grover search, 2023.

\bibitem[RS16]{RS16}
Neil~J. Ross and Peter Selinger.
\newblock Optimal ancilla-free clifford+t approximation of z-rotations.
\newblock {\em Quantum Info. Comput.}, 16(11–12):901–953, September 2016.

\bibitem[SBW{\etalchar{+}}21]{SBWRB21}
Yuan Su, Dominic~W. Berry, Nathan Wiebe, Nicholas Rubin, and Ryan Babbush.
\newblock Fault-tolerant quantum simulations of chemistry in first quantization.
\newblock {\em PRX Quantum}, 2:040332, Nov 2021.

\bibitem[Sel15]{Sel15}
Peter Selinger.
\newblock Efficient clifford+t approximation of single-qubit operators.
\newblock {\em Quantum Info. Comput.}, 15(1–2):159–180, January 2015.

\bibitem[TT23]{TT23}
Ewin Tang and Kevin Tian.
\newblock A cs guide to the quantum singular value transformation, 2023.

\end{thebibliography}

\appendix
\section{Proof of \texorpdfstring{\Cref{lem:telescope_operator_norm_ineq}}{Lemma 2.1}}
\label{app:supp_proof}

\lemcompositionerrorbound*

\begin{proof}
We will prove the lemma by induction on $L$.

The base case of $L=1$ is trivial. 
For the induction step, assume the claim holds for a product of $L-1$ terms for some $L\geq 2$, that is  
$W = (V_1\otimes I_{2^{m-m_1}})\cdots (V_{L-1}\otimes I_{2^{m-m_{L-1}}}) \in \U(2^{n+m})$ implements $U = U_1 \cdots U_{L-1}$ to error $(L-1)\epsilon$, i.e.\
\begin{equation}\label{eq:telescope_induction_hypo}
    \norm{W\cdot (I_{2^n} \otimes \ket{0^m}) - U_1\cdots U_{L-1}\otimes \ket{0^m}} \leq (L-1)\epsilon. 
\end{equation}
Without loss of generality, we can set $m = \max_{i\in [L]} m_i$. 
The goal is to show that when we add one more term of $V_L\otimes I_{2^{m-m_L}}$ to $W$, for any $n$-qubit state $\ket{\psi}$, we have that
\begin{equation}\label{eq:telescope_goal}
    \norm{W\cdot (V_{L}\otimes I_{2^{m-m_{L}}}) \cdot (\ket{\psi} \otimes \ket{0^m}) - U_1 \cdots U_{L}\ket{\psi} \otimes  \ket{0^m}} \leq L\epsilon. 
\end{equation}
Let us adopt a telescoping sum argument to the left-hand side of \Cref{eq:telescope_goal}. 
\begin{align*}
    &\norm{W(V_{L}\otimes I_{2^{m-m_{L}}}) \cdot (\ket{\psi} \otimes \ket{0^m}) - U_1 \cdots U_{L}\ket{\psi} \otimes  \ket{0^m}} \\
    = \ & \norm{W (V_{L}\otimes I_{2^{m-m_{L}}}) \cdot (\ket{\psi} \otimes \ket{0^m}) - W (U_{L}\ket{\psi} \otimes \ket{0^m}) + W (U_{L}\ket{\psi} \otimes \ket{0^m}) - U_1 \cdots U_{L}\ket{\psi} \otimes  \ket{0^m}} \\
    \leq \ & \underbrace{\norm{W (V_{L}\otimes I_{2^{m-m_{L}}}) \cdot (\ket{\psi} \otimes \ket{0^m}) - W (U_{L}\ket{\psi} \otimes \ket{0^m})}}_{\coloneqq S_1} + \underbrace{\norm{W (U_{L}\ket{\psi} \otimes \ket{0^m}) - U_1 \cdots U_{L}\ket{\psi} \otimes  \ket{0^m}}}_{\coloneqq S_2},
\end{align*}
where the last inequality follows from the triangle inequality. We now bound $S_1$ and $S_2$ separately. 
\begin{align*}
    S_1 &\leq \norm{W}\cdot \norm{(V_{L}\otimes I_{2^{m-m_{L}}}) \cdot (\ket{\psi} \otimes \ket{0^m}) - U_{L}\ket{\psi} \otimes \ket{0^m}} 
    \tag{$\norm{AB}\leq \norm{A}\cdot \norm{B}$} \\
    &\leq \norm{(V_{L}\otimes I_{2^{m-m_{L}}}) \cdot (\ket{\psi} \otimes \ket{0^m}) - U_{L}\ket{\psi} \otimes \ket{0^m}} \tag{$\norm{U}=1$ for any unitary $U$} \\
    &= \norm{V_L\cdot (\ket{\psi} \otimes \ket{0^{m_L}}) - U_L\ket{\psi} \otimes \ket{0^{m_L}}} \\
    &\leq \epsilon, \tag{$V_L$ implements $U_L$ to error $\epsilon$ and \Cref{def:cliffordT_approx}}
\end{align*}
where the last inequality is because $V_L\in \U(2^{n+m_L})$ implements $U_L$ to error $\epsilon$. 
\begin{align*}
    S_2 &= \norm{W (\ket{\psi'} \otimes \ket{0^m}) -  U_1\cdots U_{L-1}\ket{\psi'}\otimes \ket{0^m}} \tag{set $\ket{\psi'} \coloneqq U_{L}\ket{\psi}$} \\
    &\leq (L-1)\epsilon. \tag{induction hypothesis in \Cref{eq:telescope_induction_hypo}}
\end{align*}
Hence $S_1+S_2\leq L\epsilon$ which proves \Cref{eq:telescope_goal} and thus completes the induction. 
\end{proof}

\end{document}